\documentclass[12pt]{article}
\usepackage[utf8]{inputenc}
\usepackage{amssymb}
\usepackage{makeidx}
\usepackage[english]{babel}
\usepackage{graphicx}
\usepackage{amsfonts,amsmath,amssymb,amsthm}
\usepackage{oldgerm}
\usepackage{mathrsfs}
\usepackage[active]{srcltx}
\usepackage{verbatim}
\usepackage[toc,page]{appendix}
\usepackage{aliascnt}
\usepackage{array}
\usepackage{hyperref}
\usepackage[textwidth=4cm,textsize=footnotesize]{todonotes}
\usepackage{xargs}
\usepackage{cellspace}
\usepackage[Symbolsmallscale]{upgreek}
\usepackage{geometry}
\usepackage{hyperref}
\usepackage{array}
\usepackage{fancyhdr}
\usepackage{graphicx}
\usepackage{caption}
\usepackage{subcaption}
\graphicspath{{figures/}}
\usepackage{enumerate}
\usepackage{xcolor}
\usepackage{algorithm,algorithmic}

\newcommand{\rmd}{\mathrm{d}}
\newcommand{\eqsp}{\;}
\newcommand{\1}{\mathrm{1}}

\newcommand{\E}{\mathbb{E}}
\newcommand{\qk}{q_{k}}

\newcommand{\mw}{\mathsf{w}}
\newcommand{\U}{\mathsf{U}}
\newcommand{\Lo}{\mathsf{L}}
\newtheorem{lemma}{Lemma}
\newtheorem{proposition}{Proposition}
\newcommand{\hQ}{\widehat{Q}}
\newcounter{hypA}
\newenvironment{hypA}{\refstepcounter{hypA}\begin{itemize}
\item[{\bf H\arabic{hypA}}]}{\end{itemize}}
\newcounter{defcounter}
\newenvironment{assumpt}{%
\addtocounter{equation}{-1}
\refstepcounter{defcounter}

\begin{equation}}
{\end{equation}}
\begin{document}

\author{Pierre Gloaguen\footnotemark[1] \and Marie-Pierre Etienne\footnotemark[1] \and Sylvain Le {C}orff\footnotemark[2]}
 
\footnotetext[1]{AgroParistech, UMR MIA 518, F-75231 Paris, France.}
\footnotetext[2]{Laboratoire de Math\'ematiques d'Orsay, Univ. Paris-Sud, CNRS, Universit\'e Paris-Saclay.}

\title{Online Sequential Monte Carlo smoother for partially observed stochastic differential equations}

\lhead{Gloaguen et al.}
\rhead{Particle smoother for SDE}
\date{}

\maketitle

\begin{abstract}
This paper introduces a new algorithm to approximate smoothed additive functionals for partially observed stochastic differential equations. This method relies on the recent procedure introduced in \cite{olsson:westerborn:2016} which allows to compute such approximations online, i.e. as the observations are received, and with a computational complexity growing linearly with the number of Monte Carlo samples. The algorithm of \cite{olsson:westerborn:2016} cannot be used in the case of partially observed stochastic differential equations since the transition density of the latent data is usually unknown. We prove that a similar algorithm may still be defined for partially observed continuous processes by replacing this unknown quantity by an unbiased estimator obtained for instance using general Poisson estimators. We prove that this estimator is consistent and its performance are illustrated using data from two models. 
\end{abstract}


\section{Introduction}
This paper introduces a new algorithm to solve the smoothing problem for partially observed continuous time stochastic processes. In this setting, the hidden state process $(X_t)_{t\ge 0}$ is assumed to be a solution to a stochastic differential equation (SDE) and the only information available is given by noisy observations $(Y_{k})_{0\le k\le n}$ of the states $(X_k)_{0\le k\le  n}$ at some discrete time points $(t_k)_{0\le k\le n}$. The bivariate stochastic process $\{(X_{k},Y_{k})\}_{0\le k\le n}$ is a state space model such that conditional on the state sequence $(X_{k})_{0\le k\le n}$ the observations $(Y_{k})_{0\le k\le n}$ are independent and for all $0\le \ell\le n$ the conditional distribution of $Y_{\ell}$ given $\{X_{k}\}_{0\le k\le n}$ depends on $X_{\ell}$ only.

Statistical inference for partially observed state sequences often requires to solve bayesian filtering and smoothing problems, i.e.\ the computation of the posterior distributions of sequences of hidden states given observations. The filtering problem refers to the estimation, for each $0\le k \le n$,  of the distributions of the hidden state $X_k$ given the observations $(Y_0,\ldots,Y_k)$. Smoothing stands for the estimation of the distributions of the sequence of states $(X_{k},\ldots,X_{p})$ given observations $(Y_{0},\ldots,Y_{\ell})$ with $0\le k\le p \le \ell \le n$. 
These posterior distributions are crucial to compute maximum likelihood estimators of unknown parameters using the observations $(Y_0,\ldots,Y_n)$ only. For instance, the E-step of the EM algorithm introduced in \cite{dempster:laird:rubin:1977}  boils down to the computation of a conditional expectation of an additive functional of the hidden states given all the observations up to time $n$. Similarly, by Fisher's identity, recursive maximum likelihood estimates may be computed using the gradient of the loglikelihood which can be written as the conditional expectation of an additive functional of the hidden states.
 See \cite[Chapter $10$ and $11$]{cappe:moulines:ryden:2005}, \cite{kantas:doucet:signh:2015,lecorff:fort:2013a,lecorff:fort:2013b,doucet:poyiadjis:singh:2011}
for further references on the use of these smoothed expectations of additive functionals  applied to maximum likelihood parameter inference in latent data models.

The exact computation of these expectations is usually not possible in the case of partially observed diffusions. In this paper, we propose to use Sequential Monte Carlo (SMC) methods to approximate smoothing distributions with random particles associated with importance weights.
 \cite{gordon:salmond:smith:1993,kitagawa:1996} introduced the first particle filters and smoothers for state space models by combining importance sampling steps to propagate particles with resampling steps to duplicate or discard particles according to their importance weights. Unfortunately, these methods cannot be applied directly to partially observed stochastic differential equations since some elementary quantities, such as transition densities of the hidden states, are not available explicitly. Discretization procedures may be used to approximate transition densities, for instance the Euler-Maruyama method, the Ozaki discretization which proposes a linear approximation of the drift coefficient between two observations \cite{ozaki:1992,shoji:1998}, or Gaussian based approximations using Taylor expansions of the posterior mean and variance of an observation given the observation at the previous time step, \cite{kessler:1997,kessler:lindner:sorensen:2012,uchida:yoshida:2012}. Other approaches based on Hermite polynomials expansion were also introduced by \cite{ait-sahalia:1999,ait-sahalia:2002,ait-sahalia:2008} and extended in several directions recently, see \cite{li:2013} and all the references on the approximation of transition densities therein. However, even the most recent discretization based approximations of the transition densities induce a systematic bias of particle based approximations of posterior distributions, see for instance \cite{delmoral:jacod:protter:2001}. To overcome this difficulty, \cite{fearnhead:papaspiliopoulos:roberts:2008} proposed to solve the filtering problem by combining SMC methods with an unbiased estimate of the transition densities based on the generalized Poisson estimator (GPE).  In this case, only the Monte Carlo error has to be controlled as there is no Taylor expansion to approximate unknown transition densities.
 
 The only solution to solve the smoothing problem for partially observed SDE using SMC methods has been proposed in \cite{olsson:strojby:2011} and extends the fixed-lag smoother of \cite{olsson:cappe:douc:moulines:2008}. 
 Using forgetting properties of the hidden chain, the algorithm improves the performance of \cite{fearnhead:papaspiliopoulos:roberts:2008} to approximate smoothing distributions but at the cost of a bias that does not vanish as the number of particles grows to infinity.
In the case of discrete time state space models, approximations of the smoothing distributions may also be obtained using the Forward Filtering Backward Smoothing algorithm (FFBS) and  the Forward Filtering Backward Simulation algorithm (FFBSi) developed respectively in \cite{kitagawa:1996,huerzeler:kunsch:1998,doucet:godsill:andrieu:2000} and \cite{godsill:doucet:west:2004}. 
Both algorithms require first a forward pass which produces a set of particles and weights approximating the sequence of filtering distributions up to time $n$. Then, a backward pass is performed to compute new weights (FFBS) or sample trajectories (FFBSi) in order to approximate the smoothing distributions. Recently, \cite{olsson:westerborn:2016} proposed a new SMC algorithm, the particle-based rapid incremental smoother (PaRIS), to approximate on-the-fly (i.e.\ using the observations as they are received) smoothed expectations of additive functionals. 
Unlike the FFBS algorithm, the complexity of this algorithm grows only linearly with the number of particles $N$ and contrary to the FFBSi algorithm, no backward pass is required. 

In this paper, we extend the use of  PaRIS algorithm to partially observed SDE. 
The proposed algorithm allows to approximate smoothed expectations of additive functionals online and with a complexity growing only linearly with the number of particles. The crucial and simple result (Lemma~\ref{lem:AR:unbiased}) of the application of PaRIS algorithm to SDE is that the accept reject mechanism introduced in \cite{douc:garivier:moulines:olsson:2011} ensuring the linear complexity of the procedure is still correct when the transition densities are replaced by unbiased estimates. The usual FFBS and FFBSi algorithms may not be extended this easily since they both require the computation of weights defined as ratios involving the transition densities, thus replacing these unknown quantities by unbiased estimates does not lead to unbiased estimators of the weights. The proposed Generalized Random version of PaRIS algorithm, hereafter named GRand PaRIS algorithm, may be applied to general hidden Markov models whose Markovian dynamics is ruled by a stochastic differential equation (one of the first two domains defined in \cite{beskos:papaspiliopoulos:roberts:fearnhead:2006}) but also to any general state space model where the transition density of the hidden chain may be estimated unbiasedly.

Section~\ref{sec:rwparis} describes the proposed algorithm to approximate smoothed additive functionals using unbiased estimates of the transition density of the hidden states and details the application of this algorithm when the transition density may be approximated using a GPE. 
In Section~\ref{sec:convergence}, classical convergence results for SMC smoothers are extended to the setting of this paper and illustrated with numerical experiments in Section~\ref{sec:exp}. 
All proofs are postponed to Appendix~\ref{sec:append:proofs}.

\section{The Generalized Random PaRIS algorithm}
\label{sec:rwparis}
$(X_t)_{t\ge 0}$ is defined as a weak solution to the following SDE in $\mathbb{R}^d$:
\begin{equation}
\label{eq:target:sde}
X_0 = x_0\quad\mbox{and}\quad \rmd X_t = \alpha(X_t)\rmd t + \rmd W_t\eqsp,
\end{equation}
where $(W_t)_{t\ge 0}$ is a standard Brownian motion. It is assumed that $\alpha$ is of the form $\alpha(x) = \nabla_x A(x)$ where $A: \mathbb{R}^d \to \mathbb{R}$ is a twice continuously differentiable function. The solution to \eqref{eq:target:sde} is supposed to be partially observed at times $t_0=0,\dots,t_n$ through an observation process $(Y_k)_{0\le k \le n}$ in $(\mathbb{R}^m)^{n+1}$. 
For all $0\le k \le n$, the distribution of $Y_k$ given $X_k:= X_{t_k}$ has a density with respect to a reference measure $\lambda$ on $\mathbb{R}^m$ given by $g(X_k,\cdot) = g_k(X_k)$. 
The distribution of $X_0$ has a density with respect to a reference measure $\mu$ on $\mathbb{R}^d$ given by $\chi$.
For all $0\le k \le n-1$, the conditional distribution of $X_{k+1} $ given $X_{k}$ has a density $\qk(X_{k},\cdot)$ with respect to $\mu$.
 
Let $0 \leq k \leq k' \leq n$, the joint smoothing distributions of the hidden states are defined, for all measurable function $h$ on $(\mathbb{R}^d)^{k'-k + 1}$, by:
\[
\phi_{k:k'|n}[h] = \mathbb{E}\left[h(X_k,\ldots,X_{k'})\middle|Y_{0:n}\right]\eqsp.
\]
For all $0\le k\le n$, $\phi_{k} = \phi_{k:k|k}$ denote the filtering distributions. The aim of this section is to detail the extension of PaRIS algorithm to approximate expectations of the form
\begin{equation}
\label{def:addfunc}
\phi_{0:n\vert n}[H_{n}] = \mathbb{E}\left[H_n(X_{0:n})\middle|Y_{0:n}\right] \text{ where } H_n=\sum_{k=0}^{n-1}h_k(X_k,X_{k+1})\eqsp,
\end{equation}
when the transition density of the hidden states is not available explicitly and where $\{h_k\}_{k=0}^{n-1}$ are given functions on $\mathbb{R}^d\times \mathbb{R}^d$. 
The algorithm is based on the following link between the filtering and smoothing distributions for additive functionals, see \cite{olsson:westerborn:2016}:
\begin{equation}
\phi_{0:n|n}[h] = \phi_n[T_n[h]]\eqsp,\;\mbox{where}\; T_n[h](X_n) = \E\left[h(X_{0:n})\vert X_n,Y_{0:n}\right]\eqsp.\label{eq:FFbsm:equality}
\end{equation}
The approximation of \eqref{eq:FFbsm:equality} requires first to approximate the sequence of filtering distributions. 
Sequential Monte Carlo methods provide an efficient and simple solution to obtain these approximations using sets of particles $\{\xi^{\ell}_k\}_{\ell=1}^N$ associated with weights $\{\omega^{\ell}_k\}_{\ell=1}^N$, $0\le k \le n$.

At time $k = 0$, $N$ particles $\{\xi^{\ell}_0\}_{\ell=1}^N$ are sampled independently according to  $\xi^{\ell}_0 \sim \eta_0$, where $\eta_0$ is a probability density with respect to $\mu$. 
Then, $\xi^{\ell}_0$ is associated with the importance weights $\omega_0^{\ell} = \chi(\xi^{\ell}_0)g_0 (\xi^{\ell}_0)/\eta_0(\xi^{\ell}_0)$. 
For any bounded and measurable function $h$ defined on $\mathbb{R}^d$, the expectation $\phi_{0}[h] $ is approximated by
\[
\phi^N_{0}[h] = \frac{1}{\Omega_0^N} \sum_{\ell=1}^N \omega_0^{\ell} h \left(\xi^{\ell}_0 \right)\eqsp, \quad \Omega_0^N:= \sum_{\ell=1}^N \omega_0^{\ell}\eqsp.
\]
Then, for $1\le k \le n$, using $\{(\xi^{\ell}_{k-1},\omega^{\ell}_{k-1})\}_{\ell=1}^N$, the auxiliary particle filter of \cite{pitt:shephard:1999} samples pairs $\{(I^{\ell}_k,\xi^{\ell}_{k})\}_{\ell=1}^N$ of indices and particles using an instrumental transition density $p_k$ on $\mathbb{R}^d\times \mathbb{R}^d$ and an adjustment multiplier function $\vartheta_k$ on $\mathbb{R}^d$. Each new particle $\xi^{\ell}_{k}$ and weight $\omega^{\ell}_k$ at time $k$ are computing following these steps:
\begin{enumerate}[-]
\item choose a particle index $I^{\ell}_k$ at time $k-1$ in $\{1,\ldots,N\}$ with probabilities proportional to $\omega_{k-1}^{j} \vartheta_k (\xi^{j}_{k-1})$, for $j$ in $\{1,\ldots,N\}$ ;
\item sample  $\xi^{\ell}_{k}$ using this chosen particle according to $\xi^{\ell}_{k} \sim p_k(\xi^{I^{\ell}_k}_{k-1},\cdot)$ ; 
\item  associate the particle $\xi^{\ell}_k$ with the importance weight:
\begin{equation}
\label{eq:importance:weights}
\omega^{\ell}_k := \frac{\qk(\xi_{k-1}^{I^{\ell}_k},\xi^{\ell}_k)g_k(\xi^{\ell}_k)}{\vartheta_k(\xi^{I^{\ell}_k}_{k-1}) p_k (\xi_{k-1}^{I^{\ell}_k},\xi^{\ell}_k)}\eqsp.
\end{equation}
\end{enumerate} 
The expectation $\phi_{k}[h]$ is approximated by
\[
\phi^N_{k}[h] := \frac{1}{\Omega_k^N} \sum_{\ell=1}^N \omega_k^{\ell} h \left(\xi^{\ell}_k \right)\eqsp,\quad\Omega_k^N:= \sum_{\ell=1}^N \omega_k^{\ell}\eqsp.
\]
PaRIS algorithm uses the same decomposition as the FFBS algorithm introduced in \cite{doucetgodsillandrieu:2000} and the FFBSi algorithm proposed by \cite{godsill:doucet:west:2004} to approximate smoothing distributions. It combines both the forward only version of the FFBS algorithm with the sampling mechanism of the FFBSi algorithm. It does not produce an approximation of the smoothing distributions but of the smoothed expectation of a fixed additive functional and thus  may be used to approximate \eqref{def:addfunc}. Its crucial property is that it does not require a backward pass, the smoothed expectation is computed on-the-fly with the particle filter and no storage of the particles or weights is needed. 

PaRIS algorithm relies on the following fundamental property of $T_k[H_k]$ when $H_k$ is as in \eqref{def:addfunc}:
\begin{align*}
T_k[H_k](X_k) &=\mathbb{E}\left[T_{k-1}[H_{k-1}](X_{k-1}) + h_{k-1}(X_{k-1},X_k)\middle|X_k,Y_{0:k-1} \right]\eqsp,\\
&= \frac{\int \phi_{k-1}(\rmd x_{k-1})q_{k-1}(x_{k-1},X_k)\left\{T_{k-1}[H_{k-1}](x_{k-1}) + h_{k-1}(x_{k-1},X_k)\right\}}{\int \phi_{k-1}(\rmd x_{k-1})q_{k-1}(x_{k-1},X_k)}\eqsp.
\end{align*}
Therefore, \cite{olsson:westerborn:2016} introduces sufficient statistics $\tau^i_k$ (starting with $\tau^i_0 = 0$, $1\le i\le N$), approximating $T_k[H_k](\xi^i_k)$, for $1\le i\le N$ and $0\le k \le n$. First, replacing $\phi_{k-1}$ by $\phi^N_{k-1}$ in the last equation leads to the following approximation of $T_k[H_k](\xi^i_k)$:
\begin{equation}
\label{eq:Tk}
T_k^N[H_k](\xi_k^i) = \sum_{j=1}^N \Lambda_{k-1}^N(i,j)\left\{T_{k-1}[H_{k-1}](\xi_{k-1}^j) + h_{k-1}(\xi^j_{k-1},\xi^i_k)\right\}\eqsp, 
\end{equation}
where
\begin{equation}
\label{eq:Lambda}
\Lambda_{k}^N(i,\ell) = \frac{\omega^{\ell}_{k} \qk(\xi^{\ell}_{k},\xi_{k+1}^{i})}{\sum_{\ell=1}^N\omega^{\ell}_{k} \qk(\xi^{\ell}_{k},\xi_{k+1}^{i})}\eqsp,\quad 1\le \ell\le N\eqsp.
\end{equation}
Computing exactly these approximations would lead to a complexity growing quadratically with $N$ because of the normalizing constant in \eqref{eq:Lambda}. Therefore, PaRIS algorithm samples particles in the set $\{\xi^j_{k-1}\}_{j=1}^N$ with probabilities $\Lambda_{k}^N(i,\cdot)$ to approximate the expectation \eqref{eq:Tk} and produce $\tau^i_k$.
Choosing $\tilde{N}\ge 1$, at each time step $0\le k \le {n-1}$ these statistics are updated according to the following steps.
\begin{enumerate}[(i)]
\item \label{it:PaRIS:filt} Run one step of a particle filter to produce $\{(\xi^{\ell}_k, \omega^{\ell}_k)\}$ for $1\le \ell \le N$.
\item \label{it:PaRIS:sampleindex} For all $1\le i \le N$, sample independently $J_{k}^{i,\ell}$ in $\{1,\ldots,N\}$ for $1\le \ell \le \widetilde N$ with probabilities $\Lambda_{k}^N(i,\cdot)$, given by \eqref{eq:Lambda}.
\item \label{it:PaRIS:smooth} Set
\[
\tau^{i}_{k+1} := \frac{1}{\widetilde{N}} \sum^{\widetilde{N}}_{\ell=1} \left\{ \tau^{J_{k}^{i,\ell}}_{k} + h_{k} \left(\xi^{J_{k}^{i,\ell}}_{k}, \xi^{i}_{k+1}\right)  \right\}\eqsp.
\]
\end{enumerate}
Then, \eqref{def:addfunc} is approximated by
\[
\phi_{0:n\vert n}^N[\tau_n] = \frac{1}{\Omega_n^N}\sum_{i=1}^N \omega^{i}_n \tau_n^i\eqsp.
\] 
As proved in \cite{olsson:westerborn:2016}, the algorithm is asymptotically consistent (as $N$ goes to infinity) for any precision parameter $\tilde N$. However, there is a significant qualitative difference between the cases $\tilde{N} = 1$ and $\tilde{N} \geq 2$. As for the FFBSi algorithm,  when there exists $\sigma_+$ such that $0<\qk <\sigma_+$, PaRIS algorithm may be implemented with $\mathcal{O}(N)$ complexity using the accept-reject mechanism of \cite{douc:garivier:moulines:olsson:2011}.

In general situations, PaRIS algorithm cannot be used for stochastic differential equations as $\qk$ is unknown. Therefore, the computation of the importance weights $\omega_{k}^{\ell}$ and of the acceptance ratio of \cite{douc:garivier:moulines:olsson:2011} is not tractable. Following \cite{fearnhead:papaspiliopoulos:roberts:2008,olsson:strojby:2011}, filtering weights can be approximated by replacing $\qk(\xi^{\ell}_{k},\xi_{k+1}^{i})$  by an unbiased  estimator $\widehat{q}_k(\xi^{\ell}_{k},\xi_{k+1}^{i};\zeta_k)$, where $\zeta_k$ is a random variable in $\mathbb{R}^q$ such that:
\[
\widehat{q}_k(\xi^{\ell}_{k},\xi_{k+1}^{i};\zeta_k)> 0~~\text{a.s}\quad\mbox{and}\quad
\mathbb{E}\left[\widehat{q}_k(\xi^{\ell}_{k},\xi_{k+1}^{i};\zeta_k)\middle| \mathcal{G}_{k+1}^N\right] = \qk(\xi^{\ell}_{k},\xi_{k+1}^{i})\eqsp,
\]
where, for all $0\le k \le n$, 
\begin{align*}
\mathcal{F}_{k}^N &= \sigma\left\{Y_{0:k};(\xi^{\ell}_{u},\omega^{\ell}_{u},\tau^{\ell}_{u});J_{v}^{\ell,j};~1\le \ell\le N,~0\le u\le k, 1\le j \le \widetilde{N}, 0\le v< k\right\}\eqsp,\\
\mathcal{G}_{k+1}^N &= \mathcal{F}_{k}^N \vee \sigma\left\{Y_{k+1};(\xi^{\ell}_{k+1},\omega^{\ell}_{k+1});~1\le \ell\le N\right\}\eqsp.
\end{align*}
Practical choices for $\zeta_k$ are discussed below, see for instance \eqref{eq:GPE1} which presents the choice made for the implementation of such estimators in our context. In the case where $\qk$ is unknown, the filtering weights in \eqref{eq:importance:weights} then become:
\begin{equation}
\label{eq:random:weight}
\widehat{\omega}^{\ell}_k := \frac{\widehat{q}_k(\xi_{k-1}^{I^{\ell}_k},\xi^{\ell}_k;\zeta_k)g_k(\xi^{\ell}_k)}{\vartheta_k(\xi^{I^{\ell}_k}_{k-1}) p_k (\xi_{k-1}^{I^{\ell}_k},\xi^{\ell}_k)}\eqsp.
\end{equation}
Therefore, to obtain a generalized random version of PaRIS algorithm, we only need to be able to sample from the discrete probability distribution $\Lambda_k^N(i,\cdot)$ in the case of SDE based HMM. Consider the following assumption: for all  $0\leq k\leq n$, there exists a random variable $\hat{\sigma}^k_+$ measurable with respect to $\mathcal{G}_{k+1}^N$ such that,
\begin{assumpt}
\label{assumpt:global}
\mathrm{sup}_{x,y,\zeta}\;\widehat{q}_k(x,y;\zeta)\leq \hat{\sigma}^k_+\eqsp.
\end{assumpt}
\begin{lemma}
\label{lem:AR:unbiased}
Assume that $A_1$ holds for some $0\le k\le n-1$.  For all $1\le i \le N$, define the random variable $J_k^i$ as follows:
\begin{algorithmic}
\REPEAT
\STATE Sample independently $\zeta$, $U\sim \mathcal{U}[0,1]$ and $J\in\{1,\ldots,N\}$ with probabilities proportional to $\{\widehat{\omega}_{k}^1,\dots,\widehat{\omega}_{k}^N\}$.
\UNTIL{$U \leq \widehat{\qk}(\xi_{k}^J,\xi_{k+1}^i,\zeta)/\hat{\sigma}^k_+$.}
\STATE Set $J_k^i = J$.
\end{algorithmic}
Then, the conditional probability distribution given $\mathcal{G}_{k+1}^N$ of $J_{k}^{i}$ is $\Lambda_{k}^N(i,\cdot)$. 
\end{lemma}
\begin{proof}
See Appendix \ref{sec:append:proofs}.
\end{proof}
Note that Lemma~\ref{lem:AR:unbiased} still holds if assumption \eqref{assumpt:global} is relaxed and replaced by one of the two following assumptions:
\begin{assumpt}\label{assumpt:ysupport}
\mathrm{sup}_{j,y,\zeta}\;\widehat{q}_k(\xi^j_k,y,\zeta)\leq \hat{\sigma}^k_+\;.
\end{assumpt}
\vspace{-\baselineskip}
\begin{assumpt}\label{assumpt:local}
\mathrm{sup}_{i,j,\zeta}\;\widehat{q}_k(\xi^j_k,\xi^j_{k+1},\zeta)\leq \hat{\sigma}^k_+\;.
\end{assumpt}
It is worth noting that under assumptions \eqref{assumpt:global} or \eqref{assumpt:ysupport}, the linear complexity property of PaRIS algorithm still holds, whereas if only assumption \eqref{assumpt:local} holds, the algorithm has a quadratic complexity.

\subsection*{Bounded estimator of $q_k$}
For $x, y \in \mathbb{R}^d$, by Girsanov and Ito's formulas, the transition density $q_k(x,y)$ of \eqref{eq:target:sde} satisfies, with $\Delta_k = t_{k+1}-t_k$,
\begin{align*}
q_k(x,y)=\varphi_{\Delta_k}(x,y)\exp\left\lbrace A(y)-A(x)\right\rbrace \mathbb{E}_{\mathbb{W}^{x,y,\Delta_k}}\left[ \exp \left\lbrace - \int_0^{\Delta_k} \phi(\mw_s)\rmd s \right\rbrace \right]\eqsp,
\end{align*}
where $\mathbb{W}^{x,y,\Delta_k}$ is the law of Brownian bridge starting at $x$ at 0 and hitting $y$ at $\Delta_k$, $(\mw_t)_{0\leq t \leq \Delta_k}$ is such a Brownian bridge, $\varphi_{\Delta_k}(x,y)$ is the p.d.f. of a normal distribution with mean $x$ and variance $\Delta_k$, evaluated at $y$ and $\phi:\mathbb{R}^d\to\mathbb{R}$ is defined as:
\[
\phi(x) =\left(\|\alpha(x)\|^2  + \triangle A(x)\right)/2\eqsp,
\]
with $\triangle$ the Laplace operator.
Assume that there exist random variables $\Lo_\mw$ and $\U_\mw$ such that for all $0\leq s \leq \Delta_k$, $\Lo_\mw \leq \phi(\mw_s)\leq \U_\mw$.
Let $\kappa$ be a random variable taking values in $\mathbb{N}$ with distribution $\mu$ and $(U_j)_{1\le j\le \kappa}$ be independent uniform random variables on $[0,\Delta_k]$, and $
\zeta_k = \left\{\kappa,\mw,U_1,\ldots,U_\kappa\right\}\eqsp$. 
As shown in \cite{fearnhead:papaspiliopoulos:roberts:2008}, a positive unbiased estimator is given by 
\begin{multline}
\widehat{q}_k(x,y;\zeta_k) = \varphi_{\Delta_k}(x,y) \exp \left\{A(y) - A(x)\right\}\\ 
\times\mathrm{exp}\left\{-\U_\mw\Delta\right\}\frac{\Delta_k^{\kappa}}{\mu(\kappa)\kappa!}\prod_{j=1}^{\kappa}\left(\U_\mw-\phi(\mw_{U_j})\right)\eqsp.\label{eq:unbiased:q}
\end{multline}
Interesting choices of $\mu$ are discussed in \cite{fearnhead:papaspiliopoulos:roberts:2008} and we focus here on the so called GPE-1, where $\mu$ is a Poisson distribution with intensity $(\U_\mw-\Lo_\mw)\Delta_k$. In that case, the estimator \eqref{eq:unbiased:q} becomes:
\begin{equation}
\widehat{q}_{k}(x,y;\zeta_k) = \varphi_{\Delta_k}(x,y) \exp \left\{A(y) - A(x)- \Lo_\mw\Delta_k \right\}\prod_{j=1}^{\kappa}\frac{\U_\mw-\phi(\mw_{U_j})}{\U_\mw-\Lo_\mw}\eqsp.\label{eq:GPE1}
\end{equation}
On the r.h.s. of \eqref{eq:GPE1}, the product over $\kappa$ elements is bounded by 1, therefore, a sufficient condition to satisfy of the assumptions \eqref{assumpt:global}-\eqref{assumpt:local} is that the function:
\begin{align}
\rho_{\Delta_k}:\eqsp\mathbb{R}^d\times \mathbb{R}^d &\mapsto \mathbb{R}\nonumber\\
(x,y)&\mapsto \varphi_{\Delta_k}(x,y) \exp \left\{A(y) - A(x)- \Lo_\mw\Delta_k \right\}\label{eq:rho:func}
\end{align}
is upper bounded almost surely by $\hat{\sigma}^k_+$.
In particular, if $\Lo_\mw$ is bounded almost surely, \eqref{eq:rho:func} always satisfies assumption \eqref{assumpt:local} and Algorithm~\ref{alg:Ozaki:PaRIS} can be used. This condition is always satisfied for models in the domains $\mathcal{D}_1$ and $\mathcal{D}_2$ defined in \cite{beskos:papaspiliopoulos:roberts:fearnhead:2006}, i.e.\ domains for which the exact algorithms EA1 and EA2 can be used.

When \eqref{assumpt:global} or \eqref{assumpt:ysupport} holds, it can be nonetheless of practical interest to choose the bound $\hat{\sigma}^k_+$ corresponding to \eqref{assumpt:local}. Indeed, this might increase significantly the acceptance rate of the algorithm, and therefore reduce the number of drawings of the random variable $\zeta$, which has a much higher cost than the computation of $\rho$, as it requires simulations of Brownian Bridges. 
Moreover, this latter option can also avoid numerical optimization if no analytical expression of $\hat{\sigma}_+^k$ is available. In practice, we found this option more efficient in terms of computation time when $N$ has moderate values. 

\begin{algorithm}
\caption{GRand PaRIS algorithm}
\begin{algorithmic}
\FORALL{$i \in 1,\dots, N$}
\STATE Sample $\xi_0^i \sim\eta_0$, $\tau_0^i = 0$  and  $\widehat \omega_0^i = g_0(\xi_0^i)\chi_0(\xi_0^i)/\eta_0(\xi_0^i)$.
\ENDFOR
\FOR{$k \in 0,\dots, n-1$}
\FORALL{$i \in 1,\dots, N$}
\STATE Set $\tau_{k+1}^i=0$;
\STATE Sample $I_{k+1}^{i}$ in $\{1,\ldots,N\}$ with probabilities proportional to $\{\widehat{\omega}_{k}^1\vartheta_{k+1}(\xi_{k}^1),\dots,\widehat{\omega}_{k}^N\vartheta_{k+1}(\xi_{k}^N)\}$.
\STATE Sample $\xi_{k+1}^{i} \sim p_k(\xi_{k}^{I_{k+1}^{i}},\cdot)$.
\STATE For all $1\le m\le M$, sample independently $\zeta_k^m=(\kappa_m,\mw_m, (U_j^m)_{1\leq j\leq \kappa_m})$ with $\kappa_m\sim \mu$, $\mw_m\sim \mathbb{W}_k^{X_{k+1}}$ and $(U_j^m)_{1\leq j\leq \kappa_m}\sim \mathcal{U}[0,\Delta_k]^{\otimes \kappa_m}$.
\STATE Compute $\widehat{\omega}^{i}_{k+1}$ using equation \eqref{eq:random:weight}.
\FORALL{$\ell \in 1,\dots, \widetilde N$} 
\STATE Sample $J_k^{i,\ell}$ as in Lemma~\ref{lem:AR:unbiased}.
\STATE Update $\tau_{k+1}^i = \tau_{k+1}^i + (\tau^{J_k^{i,\ell}}_{k} + h_k(\xi^{J_k^{i,\ell}}_{k},\xi^i_{k+1}))/\tilde{N}$.
\ENDFOR
\ENDFOR
\ENDFOR
\end{algorithmic}
\label{alg:Ozaki:PaRIS}
\end{algorithm}

\section{Convergence results}
\label{sec:convergence}
Consider the following assumptions.
\begin{hypA}
\label{assum:boundmodel}
\begin{enumerate}[(i)]
\item For all $k \geq 0$ and  all $x\in \mathbb{R}^d$, $g_{k}(x) >0$.
\item $\underset{k\geq 0}{\sup}|g_{k}|_{\infty} < \infty$.
\end{enumerate}
\end{hypA}

\begin{hypA}
\label{assum:boundalgo}
$\underset{k\geq 1}{\sup}|\vartheta_k|_{\infty} < \infty$, $\underset{k\geq 1}{\sup}|p_k|_{\infty} < \infty$ and $\underset{k\geq 1}{\sup}|\widehat{\omega}_{k}|_{\infty} < \infty$, where
\[
\widehat{\omega}_{0}(x) = \frac{\chi(x)g_0(x)}{\eta_0(x)} \quad\mbox{and for}\; k\ge1\quad\widehat{\omega}_{k}(x,x';z) = \frac{\widehat{\qk}(x,x';z)g_{k+1}(x')}{\vartheta_{k+1}(x) p_{k} (x,x')}\eqsp.
\]
\end{hypA}

\begin{lemma}
\label{lem:iid}
For all $0\le k \le n-1$, the random variables $\{\widehat{\omega}_{k+1}^i\tau_{k+1}^i\}_{i=1}^N$ are independent conditionally on $\mathcal{F}_k^{N}$ and
\[
\mathbb{E}\left[\widehat{\omega}^1_{k+1}\tau^{1}_{k+1}\middle| \mathcal{F}_k^{N}\right] = \left(\phi^N_{k}[\vartheta_{k+1}]\right)^{-1}\phi^N_{k}\left[\int q_{k}(\cdot,x)g_{k+1}(x)\left\{\tau_k(\cdot) + h_{k+1}(\cdot,x)\right\}\rmd x\right]\eqsp.
\]
\end{lemma}

\begin{proof}
See appendix \ref{sec:append:proofs}
\end{proof}

\begin{proposition}
\label{prop:exp:deviation}
Assume that H\ref{assum:boundmodel} and H\ref{assum:boundalgo} hold and that for all $1\le k\le n$, $\mathrm{osc}(h_k)<+\infty$. For all $0\le k\le n$ and all $\widetilde{N}\ge 1$, there exist $b_k,c_k>0$ such that for all $N\ge 1$ and all $\varepsilon\in\mathbb{R}_+^\star$,
\[
\mathbb{P}\left(\left|\phi_k^N[\tau_k] - \phi_k\left[T_kh_k\right]\right|\ge \varepsilon\right)\le b_k\exp\left(-c_kN\varepsilon^2\right)\eqsp.
\]
\end{proposition}

\begin{proof}
See appendix \ref{sec:append:proofs}
\end{proof}

\section{Numerical experiments}
\label{sec:exp}
This section investigates the performance of the proposed algorithm with the sine and log-growth models. In both cases, the proposal distribution $p_k$ is chosen as the following approximation of the optimal filter (or the fully adapted particle filter in the terminology of \cite{pitt:shephard:1999}): 
$$p_k(x_{k-1},x_k)\propto \tilde{q}_k(x_{k-1}, x_k)g_k(x_{k})\eqsp,$$
where $\tilde{q}_k(x_{k-1},x_k)$ is the p.d.f. of Gaussian distibution with mean $\alpha(x_{k-1})\Delta_k$ and variance $\Delta_kI_d$, i.e.\ the Euler approximation of equation \eqref{eq:target:sde}. As the observation model is linear and Gaussian, the proposal distribution is therefore Gaussian with explicit mean and variance.

In order to evaluate the performance of the proposed algorithm, the following strategy has been chosen. We compare the estimation of the EM intermediate quantity with the one obtained by the fixed lag method of \cite{olsson:strojby:2011}, for different values of the lag (namely, 1,2,5,10,50). The particle approximation of $\mathcal{Q}(\theta,\theta)$ for each model  is computed using each algorithm, see Figure~\ref{fig:res:SINE} for the SINE model (and respectively Figure~\ref{fig:res:LG} for the log-growth model). This estimation is performed 200 times to obtain the estimates $\hQ_1,\dots,\hQ_{200}$, using $\tilde{N}=2$ particles for PaRIS algorithm, and $M=30$ replications for the Monte Carlo approximation $\widehat q_k$ of each $q_k$.  Moreover, the E step requires the computation of a quantity such as \eqref{def:addfunc} with $h_k= \log g_k + \log q_k$.  $\log q_k$ is not available explicitly and is approximated using the unbiased estimator proposed in \cite[Appendix B]{olsson:strojby:2011} based on 30 independent Monte Carlo simulations.
The intermediate quantity of the EM algorithm is also estimated                                                                                                                                                                                                                                                                                                                                                                                                                                                                                                                                                      with our algorithm 30 times using $N=5000$ particles, the reference value is then computed as the arithmetic mean of these 30 estimations, and denoted by $\hQ_\star$. Figure \ref{fig:res:SINE} (resp. \ref{fig:res:LG} ) shows this estimate for an example on one simulated data set. The GRand Paris algorithm is performed using $N=400$ particles in both cases, the fixed lag technique using $N=1600$ so that both estimations require similar computational times.

\subsection*{The SINE model} 
The performance of the GRand PaRIS algorithm are first highlighted using the SINE model, where $(X_t)_{t\geq 0}$ is supposed to be the solution to: 
\begin{equation}
\rmd X_t = \sin \left(X_t-\theta\right)\rmd t + \rmd W_t,~~X_0=x_0\eqsp. \label{eq:Lamp:SINE}
\end{equation}
This simple model has no explicit transition density, however GPE estimators may be computed by simulating Brownian bridges.
The process solution to \eqref{eq:Lamp:SINE} is observed regularly at times $t_0=0,\ldots,t_{100}=50$ through the observation process $(Y_k)_{0\leq k \leq 100}$:
\begin{equation}
Y_k = X_k + \varepsilon_k\label{eq:obs:SINE}\eqsp,
\end{equation}
where the $(\varepsilon)_{0\leq k \leq 100}$ are i.i.d. $\mathcal{N}(0,1)$.
In the example displayed on Figure \ref{fig:res:SINE}, we set $\theta=0$.
In that case, the function $\rho_{\Delta_k}$ defined in \eqref{eq:rho:func} can be upper bounded either on $(x,y)$ or only on $y$, the GRand PaRIS algorithm has therefore a linear complexity.

This same experiment was reproduced on 100 different simulated data sets.  For each simulation $s$, the empirical absolute relative bias $\mathsf{arb}_s$ and the empirical absolute coefficient of variation $\mathsf{acv}_s$ are computed as
\begin{align}
\mathsf{arb}_s &= \frac{\vert m(\hQ^s)-\hQ^s_\star\vert }{\vert \hQ^s_\star\vert }\label{eq:emp:bias}\\
\mathsf{acv}_s&=\frac{\sigma(\hQ^s)}{\vert m(\hQ^s)\vert }\label{eq:emp:cv}
\end{align}
where $m(\hQ^s)$ and $\sigma(\hQ^s)$ are the empirical mean and standard deviation of the sample $Q_1^s,\dots,Q_{200}^s$. For each estimation method, the resulting distributions of $\mathsf{arb}_1,\dots,\mathsf{arb}_{100}$ and $\mathsf{acv}_1,\dots,\mathsf{acv}_{100}$  are shown on Figure \ref{fig:mult:SINE}.\\

The GRand PaRIS algorithm outperforms the fixed lag methods for any value of the lag as the bias is the lowest (it is already negligible for $N=400$) and with a lower variance than fixed lag estimates with negligible bias (i.e., in this case, lags larger than 10). Small lags lead to strongly biased estimates for the fixed lag method, and unbiased estimates are at the cost of a large variance. It is worth noting here that the lag for which the bias is small is model dependent. 

\begin{figure}[p]
\centering
\begin{subfigure}{0.49\textwidth}
\includegraphics[width=\textwidth]{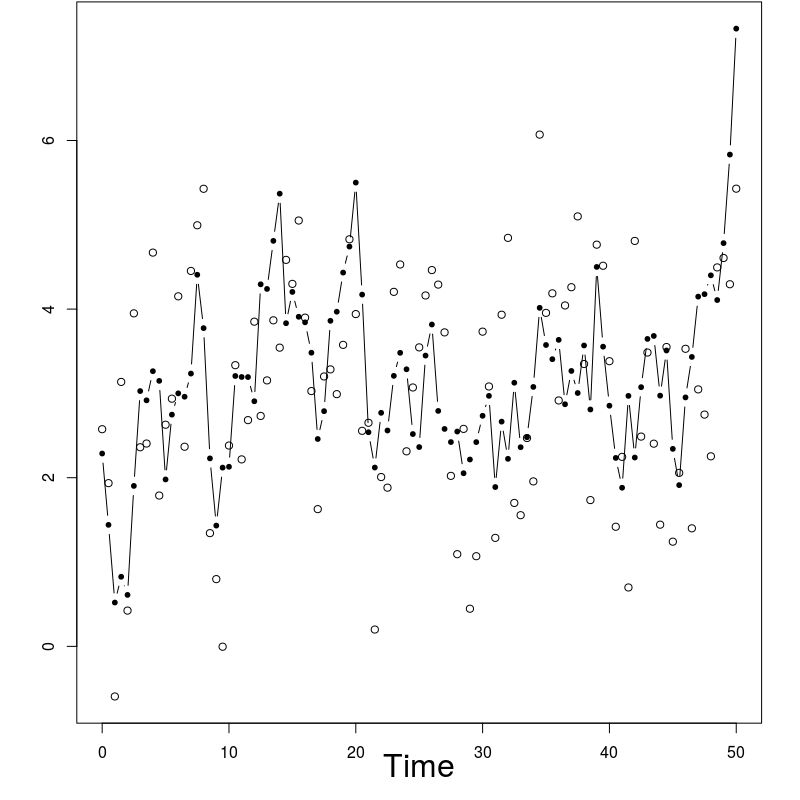}
\end{subfigure}
\begin{subfigure}{0.49\textwidth}
\includegraphics[width=\textwidth]{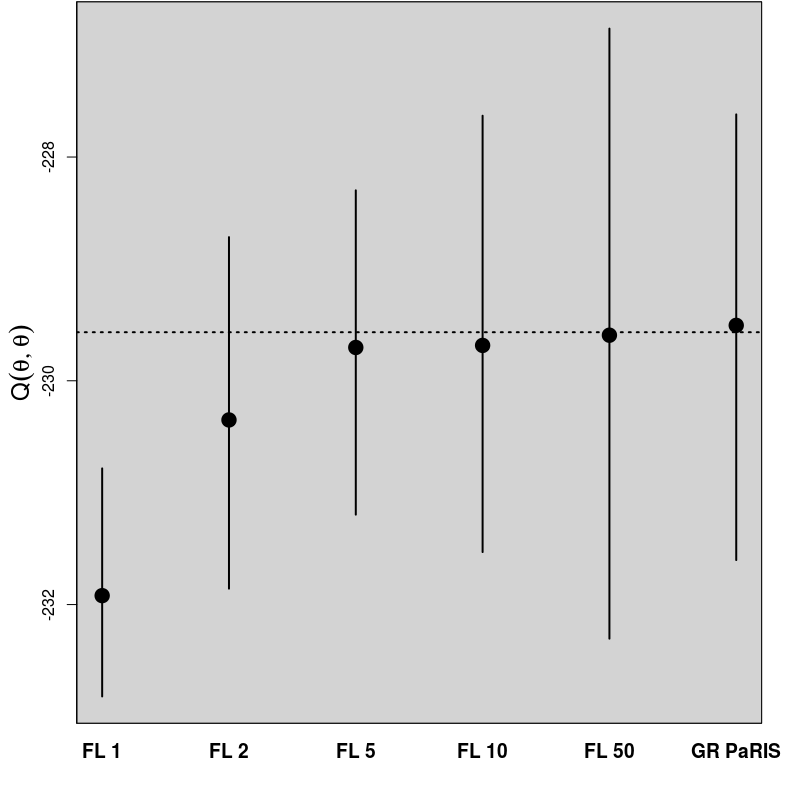}
\end{subfigure}
\caption{{\em SINE model}. Process $X$ solution to the SDE (balls) and observations $Y$ (circles) at times $t_0=0,\dots,t_{100}=50$ [left]. Estimation of the EM intermediate quantity $\mathcal{Q}(\theta,\theta)$  using the fixed lag (FL) technique for 5 different lags, and the GRand PaRIS algorithm using 200 replicates [right]. The whiskers represent the extent of the 95\% central values. The dot represents the empirical mean over the 200 replicates. The dotted line shows the reference value, computed using the GRand PaRIS algorithm with $N=5000$ particles.}
\label{fig:res:SINE}
\end{figure}
\begin{figure}[p]
\centering
\begin{subfigure}{0.49\textwidth}
\includegraphics[width=\textwidth]{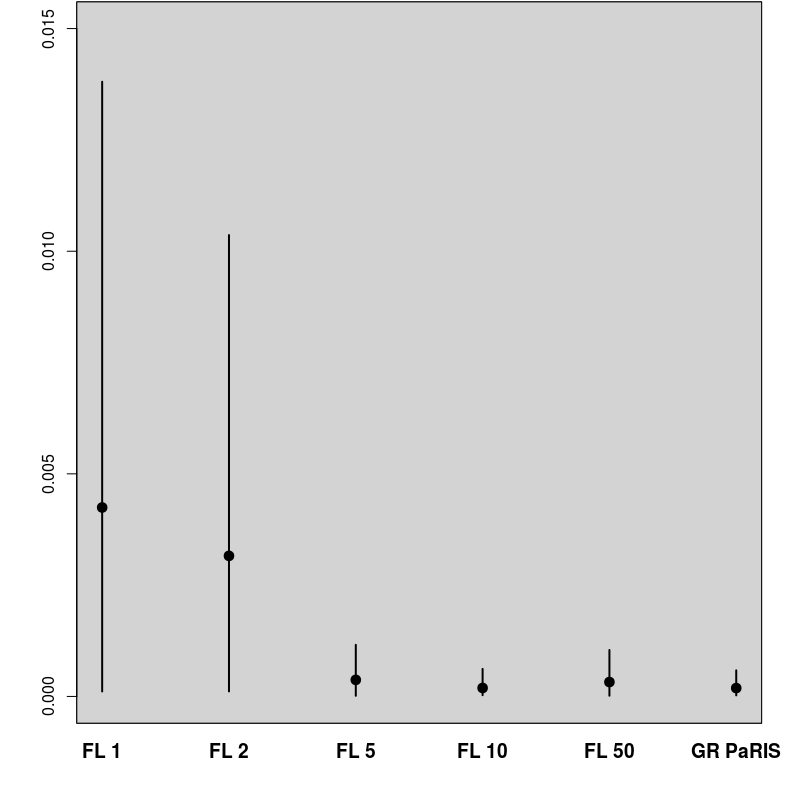}
\end{subfigure}
\begin{subfigure}{0.49\textwidth}
\includegraphics[width=\textwidth]{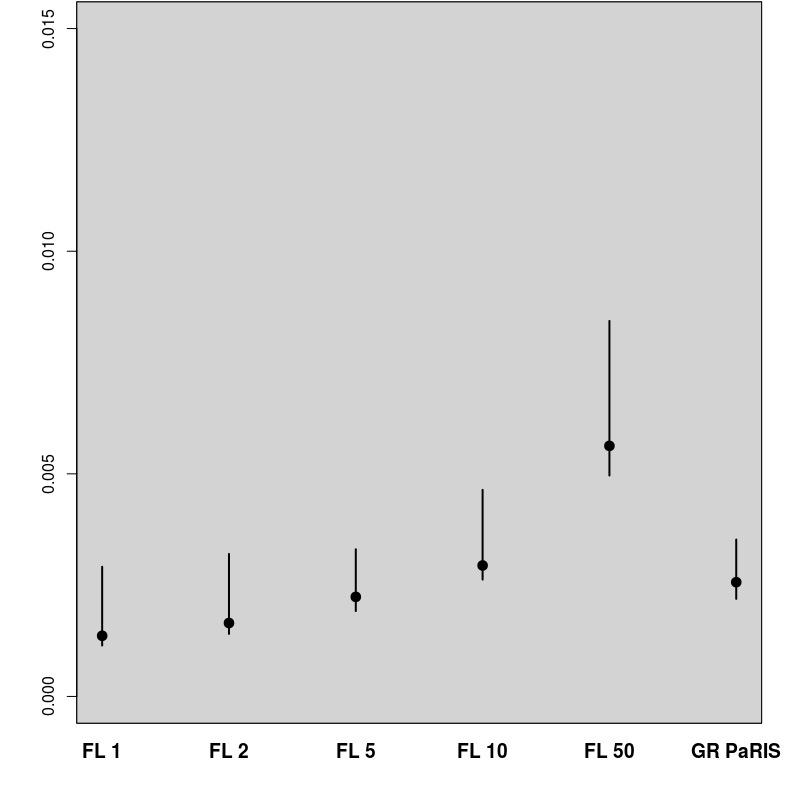}
\end{subfigure}
\caption{{\em SINE model}. Distribution of the empirical absolute relative bias [left] and of the empirical absolute coefficient of variation [right] for each method.}
\label{fig:mult:SINE}
\end{figure}

\subsection*{Log-growth model}
Following \cite{beskos:papaspiliopoulos:roberts:fearnhead:2006} and \cite{olsson:westerborn:2016}, the performance of the proposed algorithm are also illustrated with the log-growth model defined by:
\begin{equation}
\rmd Z_t = \kappa Z_t\left(1-\frac{Z_t}{\gamma}\right)\rmd t + \sigma Z_t \rmd W_t,~~Z_0=z_0. \label{eq:LG:SDE}
\end{equation}
In order to use the exact algorithms of \cite{beskos:papaspiliopoulos:roberts:fearnhead:2006} and the GPE of \cite{fearnhead:papaspiliopoulos:roberts:2008}, we consider  \eqref{eq:LG:SDE} after the Lamperti transform, i.e., the process defined by $X_t=\eta(Z_t)$, with $\eta(z) := -\log (z)/\sigma$,  which satisfies the following SDE:
\begin{equation}
\rmd X_t = \overbrace{\left( \frac{\sigma}{2} -  \frac{\kappa}{\sigma} + \frac{\kappa}{\gamma\sigma}\exp\left(-\sigma X_t\right)\right)}^{:=\alpha(X_t)}\rmd t +\rmd W_t,~~X_0=x_0=\eta(z_0).\label{eq:Lamp:LG}
\end{equation}
In this case, the conditions of the Exact Algorithm~2 defined in \cite{beskos:papaspiliopoulos:roberts:fearnhead:2006} are satisfied, as for any $m \in \mathbb{R}$ there exists $\mathsf{U}_m$ such that for all $x\ge m$, $\psi(x):=\alpha^2(x)+\alpha'(x) \leq \mathsf{U}_m$. Moreover, $\psi$ is lower bounded uniformly by  $\mathsf{L}$.  Then, GPE estimators may be computed by simulating the minimum of a Brownian bridge, and simulating Bessel bridges conditionally to this minimum, as proposed by \cite{beskos:papaspiliopoulos:roberts:fearnhead:2006}.

The process solution to \eqref{eq:Lamp:LG} is observed regularly at times $t_0=0,\dots,t_{50}=100$ through the observation process $(Y_k)_{0\le k\le 50}$ defined as:
\begin{equation}
Y_k = X_k + \varepsilon_k\label{eq:obs:LG}\eqsp,
\end{equation}
where the $(\varepsilon_k)_{0\le k \le 50}$ are i.i.d. $ \mathcal{N}(0,\sigma^2_{obs})$.
The parameters are given by $$\theta =(\kappa=0.1,\sigma=0.1,\gamma=1000,\sigma^2_{obs}=4)\eqsp.$$
In that case, the $\rho_{\Delta_k}$ function defined in \eqref{eq:rho:func} can be upper bounded as a function of $y$ when $x\in \{\xi_k^1,\dots, \xi_k^N\}$, the GRand PaRIS algorithm has therefore a linear complexity. The intermediate quantity of the EM algorithm is evaluated  as for the SINE model, see Figures~\ref{fig:res:LG} and~\ref{fig:mult:LG}.

The results for the fixed lag technique are similar to the ones presented in \cite[Figure 1]{olsson:strojby:2011} on the same model. For small lags, the variance of the estimates is small, but the estimation is highly biased. The bias rapidly decreases as the lag increases, together with a  great increase of variance.  Again, the GRand PaRIS algorithm outperforms the fixed lag smoother as it shows  a similar (vanishing) bias as the fixed lag for the largest lag and a smaller variance than the fixed lags estimates with negligible bias. 

\begin{figure}[p]
\centering
\begin{subfigure}{0.49\textwidth}
\includegraphics[width=\textwidth]{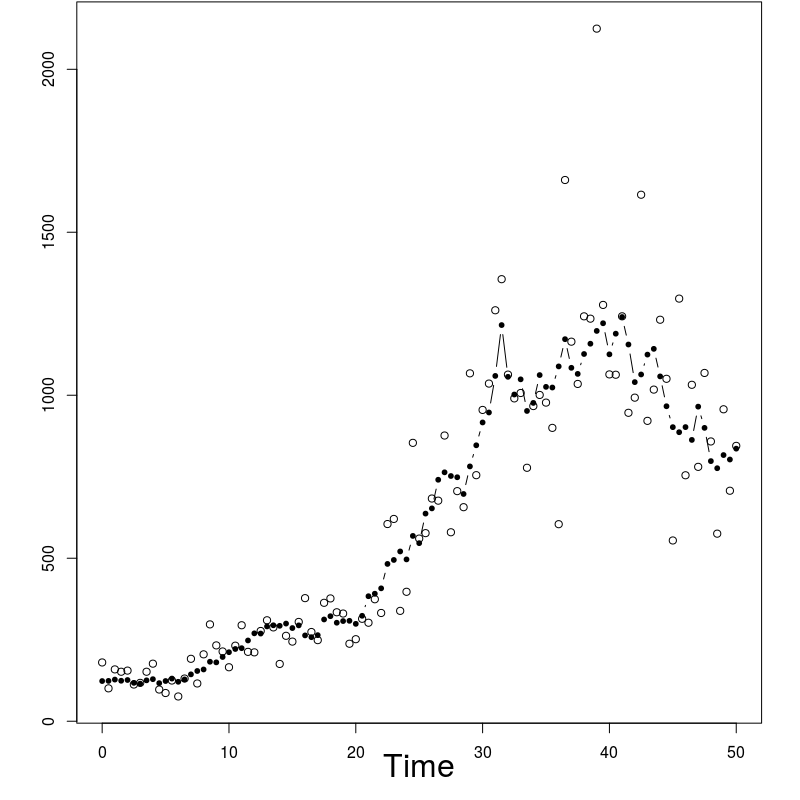}
\end{subfigure}
\begin{subfigure}{0.49\textwidth}
\includegraphics[width=\textwidth]{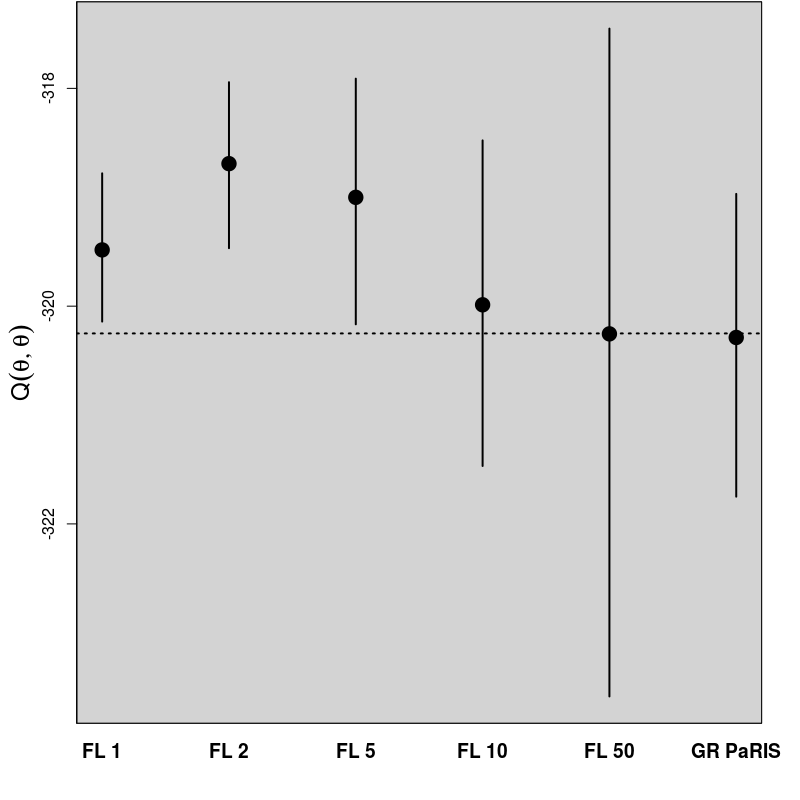}
\end{subfigure}
\caption{{\em Log-growth model}. Process $X$ solution to the SDE (balls) and observations $Y$ (circles) at times $t_0=0,\dots,t_{100}=50$ [left]. Estimation of the EM intermediate quantity $\mathcal{Q}(\theta,\theta)$  using the fixed lag (FL) technique for 5 different lags, and the GRand PaRIS algorithm using 200 replicates [right]. The whiskers represent the extent of the 95\% central values. The dot represents the empirical mean over the 200 replicates. The dotted line shows the reference value, computed using the GRand PaRIS algorithm with $N=5000$ particles.}
\label{fig:res:LG}
\end{figure}
\begin{figure}[p]
\centering
\begin{subfigure}{0.49\textwidth}
\includegraphics[width=\textwidth]{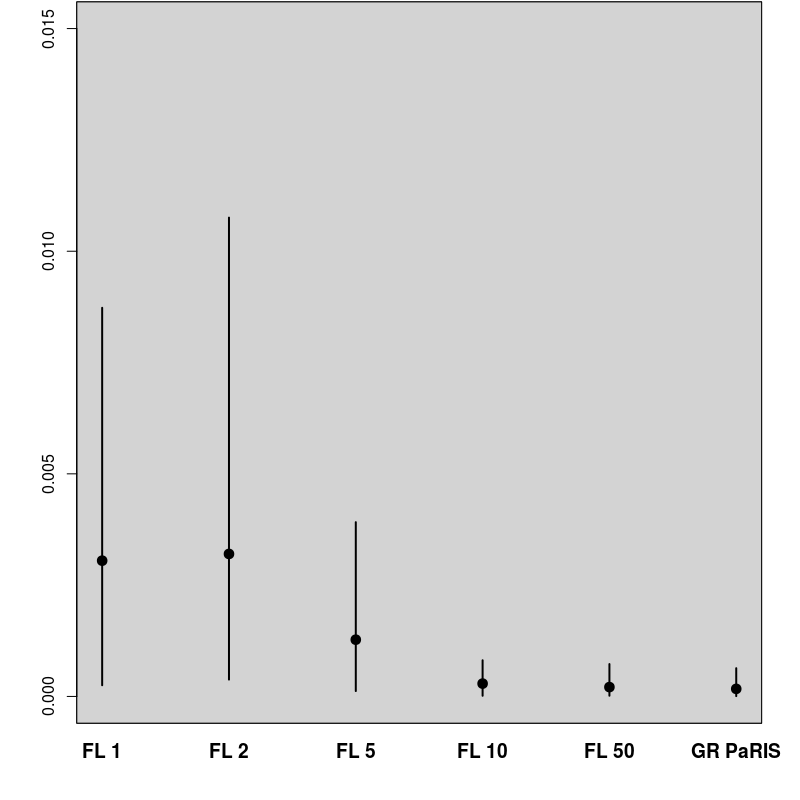}
\end{subfigure}
\begin{subfigure}{0.49\textwidth}
\includegraphics[width=\textwidth]{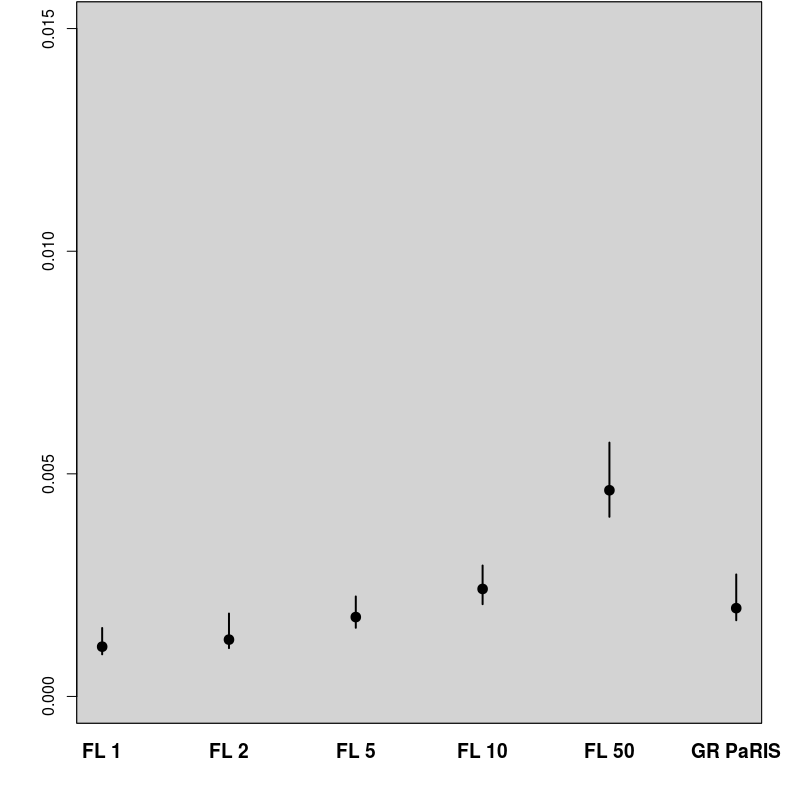}
\end{subfigure}
\caption{{\em Log-growth model}. Distribution of the empirical absolute relative bias [left] and of the empirical absolute coefficient of variation [right] for each method.}
\label{fig:mult:LG}
\end{figure}
\section{Conclusions}
This paper presents a new online SMC smoother for partially observed differential equations. This algorithm relies on an accept-reject procedure inspired from the recent PaRIS algorithm. The main result of the article for practical applications is that the mechanism of this procedure remains valid when the transition density is approximated by a an unbiased positive estimator. The proposed procedure outperforms the existing fixed lag smoother for SDE of \cite{olsson:strojby:2011}, as it does not introduce an intrinsic and non vanishing bias. In addition, numerical simulations highlight a better variance using data from two different models. It can be implemented for the class of models $\mathcal{D}_1$ and $\mathcal{D}_2$ defined in \cite{beskos:papaspiliopoulos:roberts:fearnhead:2006} with a linear complexity in $N$. 

\appendix

\section{Proofs}
\label{sec:append:proofs}
\begin{proof}[Proof of Lemma~\ref{lem:AR:unbiased}]
Let $\tau$ be the first time  draws are accepted in the accept-reject mechanism. For all $\ell\ge 1$, write
\[
\mathcal{A}^k_{\ell} = \left\{U_\ell<\widehat{q}_{k}(\xi_{k}^{J_\ell},\xi_{k+1}^{i},\zeta^{\ell}_{k})/\hat{\sigma}^k_+\right\}\eqsp.
\]
Let $h$ be a function defined on $\{1,\ldots,N\}$,
\begin{align*}
\mathbb{E}\left[h(J^{i,j}_k)\middle| \mathcal{G}_{k+1}^N\right] & = \sum_{m\ge 1}\mathbb{E}\left[h(J_m)\1_{\tau=m}\middle| \mathcal{G}_{k+1}^N\right]\eqsp,\\
& = \sum_{m\ge 1}h(m)\left(\prod_{\ell=1}^{m-1}\mathbb{E}\left[\1_{(\mathcal{A}^k_{\ell})^c}\middle| \mathcal{G}_{k+1}^N\right]\right)\mathbb{E}\left[h(J_m)\1_{\mathcal{A}^k_{m}}\middle| \mathcal{G}_{k+1}^N\right]\eqsp,\\
& = \sum_{m\ge 1}\left(\prod_{\ell=1}^{m-1}\mathbb{E}\left[1-\frac{\widehat{q_k}(\xi_{k}^{J_\ell},\xi_{k+1}^{i};\zeta_k^{\ell})}{\hat{\sigma}^k_{+}}\middle| \mathcal{G}_{k+1}^N\right]\right)\\
&\hspace{5cm}\times\mathbb{E}\left[h(J_m)\frac{\widehat{q_k}(\xi_{k}^{J_m},\xi_{k+1}^{i};\zeta_k^{m})}{\hat{\sigma}^k_{+}}\middle| \mathcal{G}_{k+1}^N\right]\eqsp,\\
& = \sum_{m\ge 1}\left(\mathbb{E}\left[1-\frac{q_k(\xi_{k}^{J_1},\xi_{k+1}^{1})}{\hat{\sigma}^k_{+}}\middle| \mathcal{G}_{k+1}^N\right]\right)^{m-1}\mathbb{E}\left[h(J_1)\frac{q_k(\xi_{k}^{J_1},\xi_{k+1}^{1})}{\hat{\sigma}^k_{+}}\middle| \mathcal{G}_{k+1}^N\right]\eqsp,\\
& = \mathbb{E}\left[h(J_1)q_k(\xi_{k}^{J_1},\xi_{k+1}^{i})\middle| \mathcal{G}_{k+1}^N\right]/\mathbb{E}\left[q_k(\xi_{k}^{J_1},\xi_{k+1}^{i})\middle| \mathcal{G}_{k+1}^N\right]\eqsp,\\
& = \sum_{\ell=1}^N \frac{h(\ell)\omega_{k-1}^{\ell}q_k(\xi_{k}^{\ell},\xi_{k+1}^{i})}{\sum_{m=1}^N\omega_{k-1}^{m}q_k(\xi_{k}^{m},\xi_{k+1}^{i})}\eqsp,\\
&= \sum_{\ell=1}^N \Lambda_{k-1}^N(i,\ell)h(\ell) \eqsp,
\end{align*}
which concludes the proof.
\end{proof}

\begin{proof}[Proof of Lemma \ref{lem:iid}]
The independence is ensured by the mechanism of SMC methods. By \eqref{eq:random:weight},
\[
\mathbb{E}\left[\widehat{\omega}^i_{k+1}\tau^{i}_{k+1}\middle| \mathcal{F}_k^{N}\right] = \mathbb{E}\left[\frac{ \widehat{\qk}(\xi_{k}^{I^{i}_{k+1}}, \xi^{i}_{k+1};\zeta_{k})g_{k+1}(\xi^{i}_{k+1})}{\vartheta_{k+1}(\xi^{I^{i}_{k+1}}_{k}) p_{k+1}(\xi_{k}^{I^{i}_{k+1}},\xi^{i}_{k+1})}\tau^{i}_{k+1}\middle| \mathcal{F}_k^{N}\right]\eqsp.
\]
Note that by Lemma~\ref{lem:AR:unbiased},
\begin{align*}
&\mathbb{E}\left[\tau^{i}_{k+1}\middle|\mathcal{G}_{k+1}^{N}\right]
 = \sum_{\ell=1}^N\frac{\omega_k^{\ell} \qk(\xi_{k}^{\ell}, \xi^{i}_{k+1}) \left(\tau^{\ell}_k + h_{k}(\xi_{k}^{\ell},\xi^{i}_{k+1})\right)}{\sum_{\ell'=1}^N\omega_k^{\ell'} \qk(\xi_{k}^{\ell'},\xi^{i}_{k+1})}\eqsp,\\
&\mathbb{E} \left[\widehat{\qk}(\xi_{k}^{I^{i}_{k+1}},\xi^{i}_{k+1};\zeta_{k}) \middle| \mathcal{G}_{k+1}^{N}\right]
 = \qk(\xi_{k}^{I^{i}_{k+1}},\xi^{i}_{k+1})\eqsp.
\end{align*}
Since $\tau^{i}_{k+1}$ and $\zeta_{k}$ are independent conditionally to $\mathcal{G}_{k+1}^{N}$:
\begin{multline*}
\mathbb{E}\left[\tau^{i}_{k+1} \widehat{\qk} (\xi_{k}^{I^{i}_{k+1}},\xi^{i}_{k+1};\zeta_{k})\middle|\mathcal{G}_{k+1}^{N}\right]\\
 = q_k(\xi_{k}^{I^{i}_{k+1}},\xi^{i}_{k+1})\sum_{\ell=1}^N\frac{\omega_k^{\ell} \qk (\xi_{k}^{\ell},\xi^{i}_{k+1})\left(\tau^{\ell}_k + h_{k}(\xi_{k}^{\ell},\xi^{i}_{k+1})\right)}{\sum_{\ell'=1}^N\omega_k^{\ell'} \qk (\xi_{k}^{\ell'},\xi^{i}_{k+1})}\eqsp.
\end{multline*}
Moreover, conditionally to $\mathcal{F}_k^N$, the probability density function of $(\xi_{k+1}^i,I_{k+1}^i)$ is given by
\[
(x,j) \mapsto \frac{\omega_k^j\vartheta_{k+1}(\xi_k^j)p_k(\xi_k^j,x)}{\Omega_k\phi_k^N[\vartheta_{k+1}]}\eqsp.
\]
Therefore, this yields:
\begin{align*}
\mathbb{E}\left[\widehat{\omega}^i_{k+1}\tau^{i}_{k+1}\middle| \mathcal{F}_k^{N}\right]&= \left(\phi^N_{k}[\vartheta_{k+1}]\right)^{-1} \sum_{j=1}^N\frac{\omega_k^j}{\Omega_k} \int \vartheta_{k+1}(\xi^{j}_{k})\frac{\qk(\xi_{k}^{j},x) g_{k+1}(x)}{\vartheta_{k+1}(\xi^{j}_{k}) p_{k}(\xi_{k}^{j},x)}\\
&\hspace{1cm}\times \sum_{\ell=1}^N\frac{\omega_k^{\ell} \qk (\xi_{k}^{\ell},x)\left(\tau^{\ell}_k + h_{k}(\xi_{k}^{\ell},x)\right)}{\sum_{\ell'=1}^N\omega_k^{\ell'}\qk(\xi_{k}^{\ell'},x)}p_{k}(\xi_{k}^{j},x)\rmd x\eqsp,\\
&= \left(\phi^N_{k}[\vartheta_{k+1}]\right)^{-1}\\
&~~~~\times\sum_{\ell=1}^N \frac{\omega_k^\ell}{\Omega_k}\left[\int \frac{ \sum_{j=1}^N \omega_k^j\qk(\xi_k^j,x) }{ \sum_{\ell'=1}^N\omega_k^{\ell'}\qk(\xi_{k}^{\ell'},x) } g_{k+1}(x)\qk (\xi_{k}^{\ell},x)\left(\tau^{\ell}_k + h_{k}(\xi_{k}^{\ell},x)\right) \rmd x \right]\\ 
& =\left(\phi^N_{k}[\vartheta_{k+1}]\right)^{-1}\phi^N_{k}\left[\int \qk(\cdot,x)g_{k+1}(x)\left\{\tau_k(\cdot) + h_{k}(\cdot,x)\right\}\rmd x\right]\eqsp,
\end{align*}
which concludes the proof.
\end{proof}

\begin{proof}[Proof of Proposition~\ref{prop:exp:deviation}]
The results is proved by induction. At time $k=0$, the result holds using that for all $1\le i \le N$, $\rho_0^i = 0$ and the convention $T_0[h_0] =0$. In addition, $\phi_0^N$ is a standard importance sampler estimator of $\phi_0$ with $\widehat \omega_0^i\le |\widehat{\omega}_0|_{\infty}$ so that for any bounded function $h$ on $\mathsf{X}$,
\[
\mathbb{P}\left(\left|\phi_0^N[h] - \phi_0\left[h\right]\right|\ge \varepsilon\right)\le b_0\exp\left(-c_0N\varepsilon^2\right)\eqsp.
\]
Assume the results holds for $k\ge 1$ and that $\vartheta_{k+1} = 1$ for simplicity. Write
\[
\phi_{k+1}^N[\tau_{k+1}] - \phi_{k+1}\left[T_{k+1}[h_{k+1}]\right] = a_N/b_N\eqsp,
\]
where $a_N = N^{-1}\sum_{i=1}^N \widehat{\omega}_{k+1}^i \left(\tau_{k+1}^i - \phi_{k+1}\left[T_{k+1}[h_{k+1}]\right]\right)$ and $b_N =N^{-1}\sum_{i=1}^N \widehat{\omega}_{k+1}^i$. By Lemma~\ref{lem:iid}, the random variables $\{\widehat{\omega}_{k+1}^i\tau_{k+1}^i\}_{i=1}^N$ are independent conditionally on $\mathcal{F}_k^{N}$ and by H\ref{assum:boundalgo},
\[
\left|\widehat{\omega}_{k+1}^i \left(\tau_{k+1}^i - \phi_{k+1}\left[T_{k+1}[h_{k+1}]\right]\right)\right| \le 2|\widehat{\omega}_{k+1}|_{\infty}|H_{k+1}|_{\infty}\eqsp.
\]
Therefore, by Hoeffding inequality,
\[
\mathbb{P}\left(\left|a_N - \mathbb{E}\left[a_N\middle|\mathcal{F}_k^{N}\right]\right|\ge \varepsilon\right) = \mathbb{E}\left[\mathbb{P}\left(\left|a_N - \mathbb{E}\left[a_N\middle|\mathcal{F}_k^{N}\right]\right|\ge \varepsilon\middle|\mathcal{F}_k^{N}\right)\right]\le 2\exp\left(-c_kN\varepsilon^2\right)\eqsp.
\] 
On the other hand,
\[
\mathbb{E}\left[a_N\middle|\mathcal{F}_k^{N}\right] = \phi^N_{k}\left[\Upsilon_k\right] \eqsp,
\]
where
\[
\Upsilon_k(x_k) = \int q_{k}(\cdot,x)g_{k+1}(x)\left(\tau_k(x_k) + h_{k+1}(x_k,x) - \phi_{k+1}\left[T_{k+1}[h_{k+1}]\right]\right)\rmd x\eqsp.
\]
By \cite[Lemma~11]{olsson:westerborn:2016}, $\phi_{k}\left[\Upsilon_k\right] = 0$ which implies by the induction assumption that 
\[
\mathbb{P}\left(\left|\mathbb{E}\left[a_N\middle|\mathcal{F}_k^{N}\right]\right|\ge \varepsilon\right)\le b_k\exp\left(-c_kN\varepsilon^2\right)\eqsp.
\]
Then,
\[
\mathbb{P}\left(\left|a_N\right|\ge \varepsilon\right) \le b_k\exp\left(-c_kN\varepsilon^2\right)\eqsp.
\] 
Similarly, as $b_N \le |\widehat{\omega}_k|_{\infty}$, by Hoeffding inequality,
\begin{multline*}
\mathbb{P}\left(\left|b_N - \mathbb{E}\left[b_N\middle|\mathcal{F}_k^{N}\right]\right|\ge \varepsilon\right) \\
= \mathbb{E}\left[\mathbb{P}\left(\left|b_N - \mathbb{E}\left[b_N\middle|\mathcal{F}_k^{N}\right]\right|\ge \varepsilon\middle|\mathcal{F}_k^{N}\right)\right]\le 2\exp\left(-c_kN\varepsilon^2\right)\eqsp.
\end{multline*}
Note that
\[
\mathbb{E}\left[b_N\middle|\mathcal{F}_k^{N}\right] = \phi^N_{k}\left[\int q_{k}(\cdot,x)g_{k+1}(x)\rmd x\right]\eqsp.
\]
By  the induction assumption,
\[
\mathbb{P}\left(\left|\mathbb{E}\left[b_N\middle|\mathcal{F}_k^{N}\right]-\phi_k\left[\int q_{k}(\cdot,x)g_{k+1}(x)\rmd x\right]\right|\ge \varepsilon\right)\le b_k\exp\left(-c_kN\varepsilon^2\right)\eqsp.
\]
The proof is completed using Lemma~\ref{lem:hoeffding:ratio}.
\end{proof}

\begin{lemma}\label{lem:hoeffding:ratio}
Assume that $a_N$, $b_N$, and $b$ are random variables defined on the same probability space such that there exist positive constants $\beta$, $B$, $C$, and $M$ satisfying
\begin{enumerate}[(i)]
    \item $|a_N/b_N|\leq M$, $\mathbb{P}$-a.s.\ and  $b \geq \beta$, $\mathbb{P}$-a.s.,
    \item For all $\epsilon>0$ and all $N\geq1$, $\mathbb{P}\left[|b_N-b|>\epsilon \right]\leq B \exp\left(-C N \epsilon^2\right)$,
    \item For all $\epsilon>0$ and all $N\geq1$, $\mathbb{P} \left[ |a_N|>\epsilon \right]\leq B \exp\left(-C N \left(\epsilon/M\right)^2\right)$.
\end{enumerate}
Then,
$$
    \mathbb{P}\left\{ \left| \frac{a_N}{b_N} \right| > \epsilon \right\} \leq B \exp{\left(-C N \left(\frac{\epsilon \beta}{2M} \right)^2 \right)} \eqsp.
$$
\end{lemma}
\begin{proof}
See \cite{douc:garivier:moulines:olsson:2011}.
\end{proof}

\bibliographystyle{plain}
\bibliography{./ParisEM_bib}
\end{document}